\documentclass[final,journal,letterpaper, 10pt]{IEEEtran}
\usepackage{bm}
\usepackage{subcaption}
\usepackage{color}
\usepackage{pdfpages}
\usepackage{cite}
\usepackage{algpseudocode}
\usepackage[linesnumbered,ruled,vlined]{algorithm2e}
\usepackage{amsmath,amsfonts,amssymb,amsthm}
\usepackage{multirow}
\usepackage{multicol}   
\usepackage[justification=centering]{caption}
 
\newtheorem{theorem}{Theorem}
\newtheorem{lemma}{Lemma} 
\newtheorem{proposition}{Proposition} 
\newtheorem{remark}{Remark}
\newtheorem{corollary}{Corollary}

\newtheorem{assumption}{Assumption}

\usepackage{siunitx}
 
\usepackage{graphicx}
\usepackage[bottom]{footmisc}

\usepackage{caption}

\usepackage{epstopdf,float,amsfonts}
\usepackage{epsfig}
\usepackage{microtype}
\usepackage{balance}

\newcommand{\Zero}{\text{$\bm{0}$}}
\newcommand{\Eye}{\text{$\bm{I}$}}
\newcommand{\Inv}[1]{\text{$#1^{-1}$}}
\newcommand{\Abs}[1]{\text{$\left|#1\right|$}}
\newcommand{\Conj}[1]{\text{$\overline{#1}$}}
\newcommand{\Real}[1]{\text{$\Re\lbrace #1\rbrace$}}
\newcommand{\Imag}[1]{\text{$\Im\lbrace #1\rbrace$}}
\newcommand{\Ker}[1]{\text{$\mathbf{ker}\left(#1\right)$}}
\newcommand{\Diff}[2]{\text{$\frac{d#1}{d#2}$}}
\newcommand{\lmax}[1]{\text{$\bm{\lambda}_{\max}\!\left(#1\right)$}}
\newcommand{\lmin}[1]{\text{$\bm{\lambda}_{\min}\!\left(#1\right)$}}
\newcommand{\Symm}[1]{\text{$\left[#1\right]_{\text{sym}}$}}
\newcommand{\AntiSymm}[1]{\text{$\left[#1\right]_{\overline{\text{sym}}}$}}
\newcommand{\Order}[1]{\text{$\mathcal{O}\left(#1\right)$}}

\newcommand{\discard}[1]{}
\newcommand{\replace}[2]{{\discard{#1}}#2}

 \title{\LARGE \bf
 Grid-Forming Storage Networks: Analytical Characterization of Damping and Design Insights}
\author{Kaustav Chatterjee, Ramij Raja Hossain, Sai Pushpak Nandanoori, Soumya Kundu, \\Subhrajit Sinha, Diane Baldwin, and Ronald Melton 
\thanks{The authors are affiliated with the Pacific Northwest National Laboratory (PNNL), Richland, WA, USA. PNNL is a multi-program national laboratory operated for the United States (US) Department of Energy (DOE) by Battelle Memorial Institute under Contract No. DE-AC05-76RL01830. This research was supported by the Embedded Storage Project funded by the US DOE Office of Electricity (OE).}
}
\begin{document}

\maketitle

\begin{abstract}
   The paper presents a theoretical study on small-signal stability and damping in bulk power systems with multiple grid-forming inverter-based storage resources. A detailed analysis is presented, characterizing the impacts of inverter droop gains and storage size on the slower eigenvalues, particularly those concerning inter-area oscillation modes. From these parametric sensitivity studies, a set of necessary conditions are derived that the design of droop gain must satisfy, to enhance damping performance. The analytical findings are structured into propositions highlighting potential design considerations for improving system stability, and illustrated via numerical studies on an IEEE 68-bus grid-forming storage network.
\end{abstract}



\section{Introduction} \label{sec_intro}

Inverter-interfaced energy storage systems are becoming increasingly important as the grid transitions to a future with high renewable penetration \cite{ibr_review, review_ibr_support_func}. These resources play a vital role in managing volatility from load-generation mismatches, providing flexibility and stability to the electric grid \cite{review_flexibility}. Traditionally, the storage resources in the grid have operated as large-scale auxiliary devices supporting ancillary services \cite{review_ibr_support_func}. 
However, with the recent advancements in power electronics designs and distributed control, there is potential for a paradigm shift \cite{gfm_review, rebalancing, distr_control_siemens}. There is a growing interest among electrical utilities and researchers, to explore grid architectures with distributed storage resources \cite{embedded_storage_quan}. Instead of a few massive centralized storage units, these approaches utilize a network of modest-sized grid-scale storage assets distributed across the power system \cite{pnnl_embedded_storage}. With suitable interfacing controls, these resources can be made sufficiently fast to respond to grid disturbances at transient time scales \cite{review_ibr_support_func}. Coordinated networks of distributed storage resources operating at the transmission-distribution boundary can act as an energy buffer between the systems, enhancing the grid's reliability multi-fold, and providing support as a core grid component --- leading to the concept of \textit{embedded storage network} \cite{pnnl_embedded_storage}. 

In recent years, grid-forming inverters (GFMs) have emerged as a promising technology for interfacing energy storage to power networks \cite{gfm_review, gfm_lasseter}. Unlike their grid-following (GFL) counterparts, GFMs do not require phase-locked loops for tracking grid voltage angles and synchronization to the grid \cite{weiDu_gfm_gfl, GFM_GFL_compare}. GFMs actively control the frequency and voltage at their point of interconnection (POI) and can emulate the dynamic behavior of synchronous generators (SGs) \cite{gfm_review}. Different types of GFMs, like virtual synchronous machines \cite{compare_vsm_droop}, droop-controlled GFMs \cite{weiDu_droop_comparison}, virtual oscillator control \cite{brian_voc}, etc., have been proposed over the years. These offer several advantages over GFLs, particularly, in frequency regulation, inertia emulation, blackout restoration, and operation of standalone microgrids \cite{gfm_review}. There is also evidence that suggests GFM-interfaced energy storage can improve oscillatory dynamics and transient behavior of power systems \cite{gfm_review}. Research has shown that tuning the controller parameters of the inverters can enhance system stability \cite{Chatterjee_def_statcom}. Particularly in droop-controlled GFMs, the choice of active power-frequency droop parameter has a direct influence on the system's eigenvalues and their damping \cite{weiDu_droop_comparison}, \cite{vsg_droop_comparison}. In addition, the inverters' available headroom, storage capacities, and grid strength also impact their dynamic performance \cite{gfm_gfl_duality, droop_small_signal}. 

Previous works \cite{weiDu_droop_comparison,  droop_small_signal, smallsignal_lowinertia, 9328533} analyzing these influences have mainly relied on numerical simulations for deriving their conclusions. Most of these studies are premised on computing the loci of the eigenvalues from batch simulations sweeping over the model parameters. 
The inferences derived, otherwise, get restricted to the specific cases and can not be generalized. Recent related efforts \cite{kundu2019identifying,nandanoori2020distributed,gorbunov2021identification,siahaan2024decentralized} have considered the problem of deriving conditions for stability as upper limits on the inverter droop gains. But these results are derived in the context of microgrids, and do not address the problem of the damping in a large-scale bulk grid network with both inverters and synchronous generators. In this paper, we attempt to address this gap. We present a theoretical analysis to compute the sensitivity of the dominant eigenvalues with changes in parameters like inverter droop setting and storage capacity. We derive the necessary conditions on the droop gain for enhancing oscillation damping. In particular, complementing the existing work \cite{kundu2019identifying,nandanoori2020distributed,gorbunov2021identification,siahaan2024decentralized} which have prescribed upper bounds on droop gains, we derive a lower bound on the choice of droop, below which the damping performance is degraded. 
The findings are summarized into stability-informed design insights.

The remainder of the paper is structured as follows. In Section \ref{sec:model} the system model is described, consisting of the reduced-order descriptions of the generators, GFM-interfaced distributed storage resources, and the power network. Section \ref{sec:analysis} presents a detailed analysis of the sensitivity of the eigenvalues on GFM droop and inverter capacity. The observations are generalized into propositions summarizing the conditions under which damping increases with a decrease in the droop setting. The numerical results from simulation studies verifying the claims on the IEEE $68$-bus system are presented in Section \ref{sec:simulation}. The concluding remarks are presented in Section \ref{sec:conclusion}. 


\section{Preliminaries}

\textbf{Notations.} The following notations will be used throughout the text. $\bm{I}_n\in\mathbb{R}^{n\times n}$ denotes the $n$-dimensional identity matrix, while $\Zero$ denotes a matrix (or, a vector) of appropriate dimensions will all entries as zeros. For every complex number $z\in\mathbb{C}$, $\overline{z}\in\mathbb{C}$ denotes its complex conjugate, with similar notation used to represent conjugate of complex vectors and matrices; while $\Real{\cdot}$ and $\Imag{\cdot}$ denote, respectively, the real and imaginary parts of a complex scalar, vector, or matrix. The notation \Abs{\cdot} is used to denote the absolute value of a complex scalar, while $\|\cdot\|_2$ is used to denote the $\ell_2$-norm of a complex vector and matrix. For any complex square matrix $C$\,, $C^H:=\Conj{C}^\top$ denotes its Hermitian (i.e., conjugate transpose). The maximum and minimum real eigenvalues of a Hermitian matrix are denoted by \lmax{\cdot} and \lmin{\cdot}\,, respectively. We use the notation 
$\Ker{\cdot}$ to denote the kernel (nullspace) of a matrix. For any Hermitian matrix $C$\,, the notations $C\succ\Zero$ and $C\prec\Zero$ (equivalently, $\succeq$ and $\preceq$) are used to denote, respectively, positive and negative definiteness (equivalently, semi-definiteness) of the matrix $C$. \replace{\input{sections/damping/cdc_damping/discard/old_properties}}{
For any $n$-dimensional square matrix $C\!\in\!\mathbb{C}^{n\times n}$\,, we denote its Hermitian and skew-Hermitian parts by $C_h$ and $C_{h'}$\,, respectively, defined as:
    \begin{align}\label{eq:def_symm}
        C_h:=\frac{1}{2}\left(C+C^H\right)\,,\,~\,C_{h'}:=\frac{1}{2}\left(C-C^H\right)\,.
    \end{align}
%
We present the following useful result: 
\begin{proposition}\label{fact3}
    Consider an $n$-dimensional complex square matrix $C$\,, and some $n$-dimensional complex vector $v$\,. Then 
    \begin{align*}
        v^HCv\text{ is real}\implies v^HCv=v^HC_hv\,.
    \end{align*} 
\end{proposition}
\begin{proof}
   The proof is trivial when we write $C=C_h+C_{h'}$\,, and notice that $v^HC_{h'}v=0$ must hold for $v^HCv$ to be real.   
\end{proof}
}

\section{System Description}\label{sec:model}

Consider a grid with $n_g$ synchronous generators (SGs) and $n_i$ grid-forming inverter (GFM)-interfaced storage resources. 
\vspace{-0.3cm}
\subsection{Synchronous Generator Model}
The dynamics of each SG $k \in \mathcal{N}_g : \{ 1, \, 2, \, \dots, n_g\}$ is represented by the second-order swing model described in (\ref{eq:SG}).
\begin{subequations} \label{eq:SG}
    \begin{align}
        \dot{\delta}_k \, &= \, \omega_k \, - \, \omega_0\\
       M_k \, \dot{\omega}_k \, + \, D_k \, (\omega_k - \omega_0) \, &= \, P^m_{k} \, - \, P_{k} 
    \end{align}
\end{subequations}
$\delta_k$ and $\omega_k$ are respectively the internal angle and speed of the $k^{\text{th}}$ machine. $\omega_0$ is the synchronous speed. $M_k$ and $D_k$ are respectively the inertia constant and the equivalent damping coefficient, and $P^m_{k}$ and $P_{k}$ are respectively the mechanical input and electrical outputs of the SG.  

\par We stack the variables from all $n_g$ SGs to define the state vectors $\delta_g := \left[ \{\delta_k \}_{k = 1}^{n_g}\right]^\top$ and $\omega_g := \left[ \{\omega_k \}_{k = 1}^{n_g}\right]^\top$, the output vector $P_g := \left[ \{P_k \}_{k = 1}^{n_g}\right]^\top$, and the parameter matrices $M := diag. \left[\{M_k\}_{k = 1}^{n_g} \right]$ and $D := diag. \left[(\{D_k\}_{k = 1}^{n_g} \right]$. Let us denote by the following the largest and smallest diagonal entries in the matrices $M$ and $D$\,:
\begin{equation}\label{eq:MD_maxmin}
    \begin{aligned}
        M_u&:=\max_{1\le k\le n_g}M_k\,,\,~\,M_l:=\min_{1\le k\le n_g}M_k\,,\\
        D_u&:=\max_{1\le k\le n_g}D_k\,,\,~\,D_l:=\min_{1\le k\le n_g}D_k\,.
    \end{aligned}
\end{equation}The state equations are next linearized about the operating point. Assuming there is no change in the mechanical power input, i.e., $\Delta P_{m,k} = 0, \, \forall \, k \in \mathcal{N}_g$, the small-signal representation of the SG dynamics may be written as follows. 
\begin{subequations} \label{eq:lin_delg_omegag}
    \begin{align}
        \Delta \dot{\delta_g} \, &= \, \Delta \omega_g \\
        \Delta \dot \omega_g &= -M^{-1}{D}\,\Delta \omega_g - M^{-1} \,\Delta {P}_g
    \end{align}
\end{subequations}

\subsection{Inverter Model}
The inverter dynamics of each storage resource $j \in \mathcal{N}_i: \{1, \, 2, \, \dots, n_i\}$ is represented by the droop-controlled GFM model described in (\ref{eq:gfm}).
\begin{subequations}
\label{eq:gfm}
\begin{align}
\dot{\delta}_j \, &= \, \omega_j \, - \omega_0\\
\dot{\omega}_j \, &= \, \frac{1}{\tau_j}\, \big(\omega_0 \, - \, \omega_j \, + \, m_{p,j} \, (P^{{set}}_j \, - \, P_{j})\big)\\
\dot{V}_j \, &= \, \frac{1}{\tau_j}\, \big(V_0 \, - \,  V_j \, + \, m_{q,j} \, (Q^{{set}}_j \, - \,  Q_{j})\big)
\end{align} 
\end{subequations}
$\delta_j$ and $\omega_j$ are the angle and frequency of the internal bus of the $j^{\text{th}}$ inverter respectively. $V_j$ is the voltage magnitude of the terminal bus to which the inverter is connected. $P^{{set}}_j$ and $Q^{{set}}_j$, and $P_j$ and $Q_j$  are respectively the real and reactive power set-points and outputs for the inverter.
$m_{p,j}$ and $m_{q,j}$ are respectively the real power-frequency and reactive power-voltage droop coefficients of the GFM normalized 
to the inverter capacity $S_j$ (in MVA), i.e., 
\begin{equation} \label{droop_def}
    \begin{aligned}
        m_{p,j} \, = \, \frac{\hat{m}_{p,j}}{S_j} \, ~~~~\text{and}~~~~
        m_{q,j} \, = \, \frac{\hat{m}_{q,j}}{S_j}\,
    \end{aligned}
\end{equation}
where, $\hat{m}_{p,j}$ and $\hat{m}_{q,j}$ are droop settings in percentages\footnote{ $\hat{m}_{p,j}$ is the real power-frequency droop setting indicating the allowable p.u. change in frequency for a 100 MW change in active power. For example, a $5\%$ droop ($\hat{m}_{p,j} = 0.05$ p.u./MW) implies for every 100 MW change, the frequency is allowed to change 5 p.u. from nominal.}. $\tau_i$ is the time constants of the low pass filter used for the active and reactive power measurements. The faster dynamics of the GFM compared to the SG impose the condition 
$\frac{1}{\tau_j} \gg \max\limits_k \frac{D_k}{M_k} \, \forall \, j \, \in \mathcal{N}_i$.   This implies,  $\tau_j\, \dot \omega_j \approx 0$ and $\tau_j\, \dot V_j \approx 0$. Therefore, while analyzing slower dynamics from the electromechanical transients it is reasonable to assume, 
\begin{subequations}
    \begin{align} 
        \omega_j \, & = \, \omega_0 \, + \, m_{p,j} \, (P^{{set}}_j \, - \, P_{j})\big)\\ 
        \implies \, \dot \delta_j \, &= \, m_{p,j} \, (P^{{set}}_j \, - \, P_{j})
    \end{align}
\end{subequations}
We stack the variables from all $n_i$ GFMs to define the state vector $\delta_i := \left[ \{\delta_j \}_{j = 1}^{n_i}\right]^\top$, the output vector $P_i := \left[ \{P_j \}_{j = 1}^{n_i}\right]^\top$, and the droop matrix $M_p := diag. \left[\{m_{p,j}\}_{j = 1}^{n_i} \right]  = m_p\, \Eye_{n_i}$. Assuming the power set-points are held constant, the linearized angle dynamics of the GFMs may be expressed as follows. 
\begin{equation} \label{eq:lin_deli}
    \Delta \dot{\delta_i} \, = \, -\, m_p \, \, \Delta P_i
\end{equation}
\subsection{Network Model}
We consider the algebraic model of a meshed transmission network where each line is represented by a $\Pi$-equivalent. The active and reactive power injections at any bus $j$ can be expressed as functions of the state variables $\delta_g$ and $\delta_i$. Following which, we may write, $P_g \, = \, f_1( \delta_g, \, \delta_i )$ and $P_i \, = \, f_2( \delta_g, \, \delta_i )$. The linearized expressions of these output variables can therefore be expressed as, 
\begin{equation} \label{eq:lin_Pg_Pi}
    \begin{aligned}
        \Delta P_g \, &= \, K_{gg} \, \Delta\delta_g \, + \, K_{gi}\, \Delta\delta_i \\
        \Delta P_i \, &= \, K_{ig} \, \Delta\delta_g \, + \, K_{ii}\,\Delta\delta_i\\
        \text{where }\,~\,K_{gg} &:= \left[ \frac{\partial P_g}{\partial \delta_g}\right]=K_{gg}^\top\,,\,K_{ii} := \left[ \frac{\partial P_i}{\partial \delta_i}\right]=K_{ii}^\top\,,\\
        K_{gi} &:= \left[ \frac{\partial P_g}{\partial \delta_i}\right]=K_{ig}^\top\,,\,K_{ig} := \left[ \frac{\partial P_i}{\partial \delta_g}\right]=K_{gi}^\top\,
    \end{aligned}
\end{equation}
are the Jacobian matrices of the partial derivatives of power injection with respect to the angle variables. $K_{gg}$ and $K_{ii}$ are diagonally dominant symmetric matrices which can be expressed as:
\begin{equation}\label{eq:Kii}\begin{aligned}
    K_{gg}&=K_{gg}^{\text{diag}} + K_{gg}^{\text{L}}\\
    K_{ii}&=K_{ii}^{\text{diag}} + K_{ii}^{\text{L}}
\end{aligned}\end{equation}
where $K_{gg}^{\text{diag}}$ and $K_{ii}^{\text{diag}}$ are diagonal matrices with positive diagonal entries, while $K_{gg}^{\text{L}}$ and $K_{ii}^{\text{L}}$ have the structure of weighted Laplacian matrices. Moreover, the matrices $K_{gg}^{\text{diag}}$ and $K_{ii}^{\text{diag}}$ satisfy the following property in typical networks:
\begin{assumption}[Network]\label{AS:network}
    The matrices $K_{gg}$, $K_{ii}$, $K_{gi}$, and $K_{ig}$ satisfy the following positive definiteness condition:
    \begin{align*}
        K_{gg}-K_{gi}\Inv{\left(K_{ii}^\top\right)}K_{ig}\succ \Zero
    \end{align*}
\end{assumption}

\begin{assumption}[Network]\label{AS:network_inverse}
    $K_{gi}K_{gi}^\top$ and $K_{ig}^\top K_{ig}$ are invertible.
\end{assumption}

\begin{remark}
    Note that for Assumption\,\ref{AS:network_inverse} to hold, we must have $n_i>n_g$\,, i.e., the number of grid-forming storage units on the network must be larger than the number of synchronous generators, which is the premise of the concept of the \textit{embedded storage network} as proposed in \cite{pnnl_embedded_storage}.
\end{remark}

\subsection{Grid Strength and Bounds on $K_{ii}$}
Note that, the diagonal entries in $K_{ii}$, denoted by $K_{ii,k}$ $\forall \, k\,\in\mathcal{N}_i$, are functions of the line admittances between the inverters' internal bus and their points of interconnection (POI). Therefore, they are an indirect measure of the grid strength \cite{gfm_gfl_duality}. Recent studies on stiff grid conditions have shown GFMs to be vulnerable to instability under high short-circuit impedance at POI \cite{gfm_gfl_duality}. Considering this, it is prudent to assume that the stability-enforcing designs would ensure the inverter admittances to be lower than a critical value. Similarly, to avoid a very weak grid scenario, the admittances must also satisfy a lower bound.
Especially, the entries of the diagonal components $K_{ii}^{\text{diag}}$ in \eqref{eq:Kii} are related to the shunt admittances connected at the POIs of the grid-forming storage units. Therefore: 
\begin{equation}\label{eq:Kii_up}
\begin{aligned}
    \gamma_l\Eye_{n_i}&\preceq K_{ii}\preceq \gamma_u\Eye_{n_i}\,,\,\\
    \text{where }\gamma_l&:=\min_{j\in\mathcal{N}_i}K_{ii,j}^{\text{diag}}>0\,,\\
    \gamma_u&:=\max_{j\in\mathcal{N}_i}\left[K_{ii,j}^{\text{diag}}+2K_{ii,j}^{\text{L}}\right]>0\,,
\end{aligned}    
\end{equation}
with $K_{ii,j}^{\text{diag}}$ and $K_{ii,j}^{\text{L}}$ denoting, respectively, the j-th diagonal entries in $K_{ii}^{\text{diag}}$ and $K_{ii}^{\text{L}}$\,.

\subsection{State-Space Representation and Oscillation Modes}
Substituting (\ref{eq:lin_Pg_Pi}) in (\ref{eq:lin_delg_omegag}) and (\ref{eq:lin_deli}), the state-space representation of the overall system can be written compactly as follows:
\begin{equation} \label{A}
\begin{aligned}
\renewcommand{\arraystretch}{1.2}
    \dot{\bm{x}} &= A\, \bm{x}\\
    \text{where,}~\bm{x}&:=\begin{bmatrix}
        \Delta \delta_g^\top & \Delta \omega_g^\top & \Delta \delta_i^\top
    \end{bmatrix}^\top\\
    \text{and }A &:= 
    \begin{bmatrix}
        \Zero & \Eye_{n_g} & \Zero \\
        -M^{-1}{K_{gg}} & -M^{-1}{D} & -M^{-1}K_{gi}\\
        -m_p\, K_{ig} &\Zero & -m_p\, K_{ii}
    \end{bmatrix}\,.
\end{aligned}
\end{equation}
The eigenvalues of $A$ characterize the natural modes of oscillation in the system. If $\lambda = \Re\{\lambda\}\,+\,j\,\Im\{\lambda\}$ is an eigenvalue of $A$, then the modal frequency, $f_\Omega$, and damping ratio, $\zeta$, for this mode are defined as follows.
\begin{subequations}\label{eq:define_damping}
    \begin{align}
        f_\Omega(\lambda) &= \frac{\Im\{\lambda\}}{2\,\pi} &\text{(in Hz)}\,,\\
        \zeta(\lambda) &= \frac{\Abs{\Re\{\lambda\}}}{|\lambda|}\,\times\, 100 &\text{(in $\%$)}
    \end{align}
\end{subequations}
This work analyzes the oscillation modes in the frequency range $0.1 < f_\Omega < 1$ Hz. These are the slower eigenvalues, commonly referred to as the inter-area modes in a power system, for which a group of generators in one area oscillates in unison against another group from other areas in the system. The reduced-order system model of (\ref{A}) is reasonably adequate for studying these slower electromechanical dynamics. 
\begin{assumption}[Slow Eigenvalues]\label{AS:slower}
    The eigenvalues $\lambda\in\mathbb{C}$ of interest of the matrix $A$, that correspond to the inter-area modes of oscillations, satisfy the condition $\Abs{\lambda}\ll m_p\left\|K_{ii}\right\|_2$\,, and that the matrix $(\lambda \Eye_{n_i} \,+\, m_p K_{ii})$ is invertible.
\end{assumption}

Next, in the paper, we perform eigenanalysis on $A$ to characterize the damping of the oscillation modes and their sensitivity to the changes in system parameters. 

\section{Main Results: Eigen-Analysis and Damping}\label{sec:analysis}

\subsection{Characterization of Slow Eigenvalues}

\begin{lemma}[Eigenvectors]\label{lem:eigenvectors}
For any eigenvalue $\lambda$ of interest of the matrix $A$ in \eqref{A}, satisfying Assumption\,\ref{AS:slower}\,, the following
\begin{align*}
    \Lambda(\lambda,m_p)\!:=\,&\lambda^2\Eye_{n_g}\!+\lambda\,\Inv{M}D\\
    &+\Inv{M}\!\left(K_{gg}\!-\!m_p\,K_{gi}\Inv{\left(\lambda\Eye_{n_i}+m_pK_{ii}\right)\!}K_{ig}\right)
\end{align*}
is a singular matrix. Moreover, for any pair of $n_g$-dimensional non-zero complex vectors $u_*,v_*\in\Ker{\Lambda(\lambda,m_p)}$\,, the pair of vectors $u,v\in\mathbb{C}^{2n_g+n_i}$ given by
\begin{align*}
    v &= \begin{bmatrix}
        \left(\lambda\,\Eye_{n_g} + M^{-1}{D}\right)\\ 
        \Eye_{n_g}\\
        -\left(\lambda\,\Eye_{n_i} + m_p\,K_{ii}\right)^{-1}K_{gi}^\top\,M^{-1}
    \end{bmatrix}\,v_*\,,\,~\\
    \text{ and }~\,u &= \begin{bmatrix}
        \Eye_{n_g} \\
        \lambda\,\Eye_{n_g}\\
        -m_p\left(\lambda\,\Eye_{n_i} + m_p~K_{ii}\right)^{-1}K_{ig}
    \end{bmatrix}\,u_*\,.
\end{align*}
are, respectively, the left and right eigenvectors associated with the eigenvalue $\lambda$ of the matrix $A$\,.    
\end{lemma}

\begin{proof} Let us consider that $v  = \begin{bmatrix}
    v_1^\top & v_2^\top & v_3^\top
\end{bmatrix}^\top$ is a left eigenvector of $A$ for the eigenvalue $\lambda$, for some pair of $n_g$-dimensional vectors $v_1,\,v_2\,\in\mathbb{C}^{n_g}$ and $n_i$-dimensional vector $v_3\in\mathbb{C}^{n_i}$\,, i.e., $v^\top  A = \lambda~ v^\top$\,. Therefore, we have:
\begin{equation}\label{eq:left_eigen}
    \begin{aligned}
        - v_2^\top\,M^{-1}K_{gg} - v_3^\top\,m_p\,K_{ig} &= \,\lambda~v_1^\top\\
        v_1^\top - v_2^\top\,M^{-1}{D} \, &= \, \lambda~v_2^\top\\
        -v_2^\top\,M^{-1}K_{gi} - v_3^\top\,m_p\,K_{ii}\, &= \, \lambda~v_3^\top
    \end{aligned}
\end{equation}
Solving the second and the third equality conditions yield:
\begin{align*}
     v_1 &= (\lambda\,\Eye_{n_g} + M^{-1}{D})\,v_2\,,\\
     v_3 &= -\left(\lambda\,\Eye_{n_i} + m_p\,K_{ii}\right)^{-1}K_{gi}^\top\,M^{-1}\,v_2\,.
\end{align*}
which leads to the aforementioned left eigenvector $v$\,. Using the above identities in the first equality condition in \eqref{eq:left_eigen}, we obtain the following:
\begin{align*}
        \Lambda(\lambda,m_p)\,v_2&=\Zero\,,
    \end{align*}
    i.e., $\Lambda(\lambda,m_p)$ must be singular, with $v_2\!\in\!\Ker{\Lambda(\lambda,m_p)}$\,. This proves the first part of the lemma. Next, considering any $v_*\in\Ker{\Lambda(\lambda,m_p)}$\,, and setting $v_2=v_*$\,, we prove that $v  = \begin{bmatrix}
    v_1^\top & v_2^\top & v_3^\top
\end{bmatrix}^\top$ is a left eigenvector of $A$\,, associated with the eigenvalue $\lambda$\,. The right eigenvector is obtained in a similar fashion, and the details are omitted for brevity.
\end{proof}

In the following, we characterize the sensitivities of these dominant modes as the penetration of grid-forming storage resources increases, i.e., we analyze how the eigenvalues $\lambda\in\mathbb{C}$ of interest of the matrix $A$ change as $m_p$ varies. In particular, the inter-area modes of oscillation of interest for this study are identified by their distinct frequencies which are largely determined by the inertia of the synchronous generators, and do not vary (significantly) with the changes in grid-forming storage droop settings and/or their size, i.e.,
\begin{assumption}[Inter-Area Modes]\label{AS:inter}
    The eigenvalues $\lambda\in\mathbb{C}$ corresponding to the inter-area modes of oscillation will satisfy:
    \begin{align*}
        \lim_{\Delta m_p\rightarrow 0^+}\Delta\Imag{\lambda} \approx 0\,.
    \end{align*}
\end{assumption}
 
\subsection{Parametric Sensitivity of Slow Eigenvalues}

%
%
\replace{\input{sections/damping/cdc_damping/discard/old_lemma_intermediate_modified}}{Let us define the following matrices:
\begin{equation}\label{eq:URQTheta}
    \begin{aligned}
        \!\!U_1&:=K_{gi}\,K_{ii}^{-2}\,K_{ig} &\in\mathbb{R}^{n_g\times n_g}\\
        \!\!U_2&:=K_{gi}\,K_{ii}^{-3}\,K_{ig} &\in\mathbb{R}^{n_g\times n_g}\\
        \!\!Q(\lambda)&:=2\lambda M+D=Q(\lambda)^\top   &\in\mathbb{C}^{n_g\times n_g}\\
        \!\!\Theta_1(\lambda)&:=U_1\left(2\,\Abs{\lambda}^2M+\lambda D\right)&\in\mathbb{C}^{n_g\times n_g}\\
        \!\!\Theta_2(\lambda)&:=\lambda\,U_1^2+2\,\Abs{\lambda}^2U_2\left(2\,\Conj{\lambda}M\!+\! D\right)&\in\mathbb{C}^{n_g\times n_g}\\
        \!\!R(\lambda, m_p)&:=K_{gi}\left(\lambda\,\Eye_{n_i}+m_pK_{ii}\right)^{-2}K_{ig}&\in\mathbb{C}^{n_g\times n_g}\\
        \!\!\Theta(\lambda,m_p)&:=\lambda\,R\left(m_p\,\Conj{R}+\Conj{Q}\right)&\in\mathbb{C}^{n_g\times n_g}
    \end{aligned}
\end{equation}
The conditions in \eqref{eq:lin_Pg_Pi} and \eqref{eq:Kii_up} imply that $U_1\!\succ\!\Zero\,,U_2\!\succ\!\Zero\,,$ and $R$ is symmetric. We establish the following relationships:
\begin{proposition}\label{prop:intermediate_result}
    For slower eigenvalues $\lambda\in\mathbb{C}$ satisfying Assumption\,\ref{AS:slower}, and for sufficiently large $m_p$\,, the matrices $R$ and $\Theta$ can be approximated as follows:
    \begin{align*}        
    \lim_{m_p\rightarrow\infty} R(\lambda,m_p)&\approx \frac{1}{m_p^3}\left(m_p\,U_1 + 2\,\Conj{\lambda}\,U_2 + \Order{1/m_p}\right)\,,\\
    \lim_{m_p\rightarrow\infty} \Theta(\lambda,m_p)&\approx \frac{1}{m_p^3}\left(m_p\,\Theta_1(\lambda) + \Theta_2(\lambda) + \Order{1/m_p}\right)
    \end{align*} 
\end{proposition}
\begin{proof}    
    See Appendix\,\ref{proof:prop_R} for details.
\end{proof}

}

We are now ready to present our main result which offers a set of necessary conditions for the enhancement of damping of the oscillatory modes of interest due to increase in the size of the grid-forming storage units.


\begin{theorem}[Damping]\label{thm:main_result}
    Consider any of the slower eigenvalues $\lambda\in\mathbb{C}$ that satisfies Assumption\,\ref{AS:slower} and Assumption\,\ref{AS:inter}. For sufficiently large $m_p$, a necessary condition for $\Real{\lambda}$ to decrease due to an incremental decrease in $m_p$ is given by:
    \begin{align}\label{eq:condition_damping}
    m_p\,\lmax{\Theta_{1,h}(\lambda)}+\lmax{\Theta_{2,h}(\lambda)}>0\,,
    \end{align}
    where $\Theta_{1,h}(\lambda)$ and $\Theta_{2,h}(\lambda)$ denote, respectively, the Hermitian components of $\Theta_1(\lambda)$ and $\Theta_2(\lambda)$\,.
\end{theorem}
\begin{proof}
Let $\Delta m_p$ be the incremental perturbation in the droop coefficient of each inverter. Consequently, the perturbation in the system matrix $A$ is expressed as
\begin{equation*}
    \Delta A \, = \, \begin{bmatrix}
        \Zero & \Zero & \Zero\\
        \Zero & \Zero & \Zero\\
        -K_{ig} & \Zero & -K_{ii}
    \end{bmatrix} \, \Delta m_p.
\end{equation*}
Let $(\lambda\,+\,\Delta\lambda)$ be the eigenvalue of the perturbed matrix $(A\,+\,\Delta A)$. The perturbation in the eigenvalues can be quantified as
\begin{equation} \label{sensitivity_1}
    \lim_{\Delta m_p\rightarrow 0^+}\Delta \lambda \, = \, \lim_{\Delta m_p\rightarrow 0^+}\frac{v^\top\, \Delta A\, u}{v^\top u}\,.
\end{equation}    
where $v$ and $u$ are the corresponding left and right eigenvectors. Using the eigenvectors from Lemma\,\ref{lem:eigenvectors}, and after some algebraic manipulation, we obtain the following:
\replace{\input{sections/damping/cdc_damping/discard/old_proof_theorem1}
}{
\begin{equation}\label{sensitivity_2}
\begin{aligned}
    \Diff{\lambda}{m_p}&:=\lim_{\Delta m_p\rightarrow 0^+}\frac{\Delta \lambda}{\Delta m_p}=\lambda\,\frac{v_*^\top\,\Inv{M}R\,u_*}{v_*^\top\Inv{M}\left(m_p\,R+Q\right)u_*}\\
    &= \lambda\,\frac{w_*^HR\,U_*\left(m_p\Conj{R}+\Conj{Q}\right)w_*}{\left|\,w_*^H\left(m_pR+Q\right)u_*\right|^2}
\end{aligned}    
\end{equation}
where we define $w_*^{}\!:=\!\Inv{M}\,\Conj{v}_*$ and $U_*\!:=\!u_*^{}u_*^H$\,. Further, recall that $R$ and $Q$ are defined in \eqref{eq:URQTheta} and that $v_*,u_*\!\in\!\Ker{\Lambda(\lambda,m_p)}$\,. The final expression in \eqref{sensitivity_2} is obtained by multiplying the numerator and the denominator by the conjugate (transpose) of the denominator.

As per Assumption\,\ref{AS:inter}, we are interested in the eigenvalues for which the imaginary parts do not change (significantly), i.e., when $d\lambda/dm_p$ is real. Applying Proposition\,\ref{fact3}, we argue:
\begin{align*}
    \Diff{\lambda}{m_p}~\text{is real}&\iff w_*^H\left[\lambda\,R\,U_*\left(m_p\Conj{R}+\Conj{Q}\right)\right]w_*~\text{is real}\\
    &\implies w_*^H\left[\lambda\,R\,U_*\left(m_p\Conj{R}+\Conj{Q}\right)\right]_hw_*~\text{is real}
\end{align*}
Moreover, $d\lambda/dm_p$ must take positive values for $\Real{\lambda}$ to decrease due to an incremental decrease in $m_p$ for every eigenvalue of interest $\lambda$ satisfying Assumption\,\ref{AS:inter}. Therefore:
\begin{align*}
    \Diff{\lambda}{m_p}>0&\implies w_*^H\left[\lambda\,R\,U_*\left(m_p\Conj{R}+\Conj{Q}\right)\right]_hw_*>0\\
    &\implies \lmax{\left[\lambda\,R\,U_*\left(m_p\Conj{R}+\Conj{Q}\right)\right]_h}>0\\
    &\implies \lmax{\Theta_h(\lambda,m_p)}>0
\end{align*}
where, the last condition follows from the positive semi-definiteness of $U_*$\,, and  $\Theta_h(\lambda,m_p)$ denotes the Hermitian component of $\Theta(\lambda, m_p)$ introduced in \eqref{eq:URQTheta}. On application of Proposition\,\ref{prop:intermediate_result}, we obtain the following necessary condition:
\begin{align*}
    \lim_{m_p\rightarrow \infty}\Diff{\lambda}{m_p}>0&\implies \lmax{m_p\Theta_{1,h}(\lambda)+\Theta_{2,h}(\lambda)}>0
\end{align*}
which completes the proof.}
\end{proof}

\subsection{Analytical Design Insights}
While Theorem\,\ref{thm:main_result} provides a necessary condition for stability, it is possible to synthesize derived conditions that provide useful design insights. Let us introduce the following:
\begin{equation}\label{eq:define_design_variables}
\begin{aligned}
    D^*(\lambda)&:={ \frac{M_u\left(1/\gamma_u+4\Abs{\lambda}^2M_l\right)}{\gamma_u/\gamma_l}}\\
    \zeta_l(\lambda)&:=\frac{2\,\Abs{\lambda}\,D_u\left(\gamma_u/\gamma_l\right)^3}{1/\gamma_u+4\,\Abs{\lambda}^2\,M_l}\\
    \zeta_u(\lambda)&:=\frac{2\,\Abs{\lambda}\,M_u\left(\gamma_u/\gamma_l\right)^2}{D_l}
\end{aligned}        
\end{equation}
where $\gamma_l,\gamma_u$ are defined in \eqref{eq:Kii_up} and $M_u,D_u$ are from \eqref{eq:MD_maxmin}.

\begin{theorem}[Design Insight]\label{thm:design}
    Consider inter-area oscillatory modes characterized in Assumption\,\ref{AS:inter} and satisfying the following condition (with $\zeta(\lambda)$ defined in \eqref{eq:define_damping}):
    \begin{align}
        D_u\,D_l&< D^*(\lambda)\,, \,\text{ and }\,~\,\zeta_l(\lambda)< \zeta(\lambda)< \zeta_u(\lambda)\,.\label{eq:design_assumptions}
    \end{align}
    Then, for the condition in Theorem\,\ref{thm:main_result} to hold, the droop control parameter $m_p$ must satisfy the following lower bound: 
    \begin{align}\label{eq:condition_design}
        m_p>m^*(\lambda):=\frac{\left(1+4\Abs{\lambda}^2M_l\,\gamma_u\right)\left(\zeta(\lambda)\!-\!\zeta_l(\lambda)\right)}{\gamma_u^2\,D_l\left(\zeta_u(\lambda)-\zeta(\lambda)\right)}\,.
    \end{align}        
\end{theorem}
\begin{proof}
    The Hermitian components of $\Theta_1$ and $\Theta_2$ are given by:
    \begin{align*}
        \Theta_{1,h}&\!=\!\Abs{\lambda}^2\left(U_1M+MU_1\right)+\frac{1}{2}\left(\lambda U_1D+\Conj{\lambda}DU_1\right)\\
        \Theta_{2,h}&\!=\!\Real{\lambda}U_1^2+\Abs{\lambda}^2\!\left[U_2\left(2\Conj{\lambda}M\!+\!D\right)+\left(2\lambda M\!+\!D\right)U_2\right]
    \end{align*}
    Using \eqref{eq:MD_maxmin} and noticing that $\lambda U_1\!+\!\Conj{\lambda}U_1\!=\!2\,\Real{\lambda}\,U_1\!\prec\! 0$ and $\lambda U_2\!+\!\Conj{\lambda}U_2\!=\!2\,\Real{\lambda}\,U_2\!\prec\! 0$ (negative definite), we obtain:
    \begin{align*}
        \lmax{\Theta_{1,h}}&\leq 2\Abs{\lambda}^2M_u\lmax{U_1}-\Abs{\Real{\lambda}}D_l\lmin{U_1}\\
        \lmax{\Theta_{2,h}}&\leq -\Abs{\Real{\lambda}}\!\left[\lmin{U_1^2}+4\Abs{\lambda}^2M_l\lmin{U_2}\right]\\
        &\,\quad\,+2\Abs{\lambda}^2D_u\lmax{U_2}
    \end{align*}
    Applying the bounds from $\eqref{eq:Kii_up}$, we have:
    \begin{align*}
        \lmax{\Theta_{1,h}}&\leq \frac{\Abs{\lambda}\,D_l\,\|K_{gi}\|_2^2}{\gamma_u^2}\left(\zeta_u(\lambda)-\zeta(\lambda)\right)\\
        \lmax{\Theta_{2,h}}&\leq -\frac{\Abs{\lambda}\,\|K_{gi}\|_2^2}{\gamma_u^4}\!\left(1+4\Abs{\lambda}^2M_l\,\gamma_u\right)\left(\zeta(\lambda)\!-\!\zeta_l(\lambda)\right)
    \end{align*}
    Note that the condition $D_u D_l\!<\!D^*(\lambda)$ in \eqref{eq:design_assumptions} ensures that $\zeta_u(\lambda)\!>\!\zeta_l(\lambda)$\,. The rest of proof follows by using the conditions \eqref{eq:design_assumptions} and the above bounds on $\lmax{\Theta_{1,h}}$ and $\lmax{\Theta_{2,h}}$ to derive necessary condition for \eqref{eq:condition_damping} (in Theorem\,\ref{thm:main_result}) to hold.
\end{proof}
\begin{remark}
Note that the lower ($\zeta_l$) and upper ($\zeta_u$) bounds on the damping ratio are not too restrictive, when one considers sufficiently low damping ($D_l, D_u$) on the system.
\end{remark}
    The condition \eqref{eq:condition_design} in Theorem\,\ref{thm:design} in  establishes system-level insights into controller parameter tuning and sizing of grid-forming storage units into bulk grid. For example, the droop controller setting (i.e., $\hat{m}_p$) must be greater than a certain value determined by a combination of the network admittance parameters $\gamma_l,\gamma_u$ which are related to the system strength at the POI, as well as system inertia ($M_l,M_u$) and damping ($D_l,D_u$) parameters. We study the dependence in a simplified scenario:
    
    \begin{corollary}\label{cor:design}
        In the limiting case of low damping ($D\rightarrow \Zero^+$) in the network, and  uniform network admittance parameters at the POI, i.e., $\gamma_l=\gamma_u$\,, we have the following:
        \begin{align*}
            \left.\lim_{D\rightarrow \Zero^+}m^*(\lambda)\right|_{\gamma_l=\gamma_u}=\frac{\left(1+4\Abs{\lambda}^2M_l\,\gamma_u\right)}{2\Abs{\lambda}M_u\,\gamma_u^2}\,\zeta(\lambda)
        \end{align*}
    \end{corollary}
    \begin{proof}
        The proof is omitted due to brevity.
    \end{proof}
    The lower bound ($m^*$) on the droop control parameter, as presented in Corollary\,\ref{cor:design}, quantifies the restrictive design requirements (higher droop settings) as the POI gets weaker (lower $\gamma_u$). The design requirements are also influenced by the available inertia in the system (with lower inertia restricting design choices). Moreover, since $m_p$ is inversely related to the inverter capacity $S_j\,\forall j\in\mathcal{N}_i$, the above observation also places an upper limit on the size of the grid-forming storage that can be connected at a given POI. Numerical results, presented next, illustrate a deterioration in the damping performance at very low droop control gains, and/or extremely large storage size.

    \begin{remark}
        Theorem\,\ref{thm:design} prescribes a lower bound on droop control gain for inter-area damping enhancement, which complements the upper limits on droop control gains prescribed in the literature \cite{kundu2019identifying,nandanoori2020distributed,gorbunov2021identification,siahaan2024decentralized}, that stem from various small-signal and nonlinear stability considerations.
    \end{remark}

\section{Numerical Results and Simulation Studies} 
\label{sec:simulation}
The positive-sequence phasor model of the IEEE $16$-machine $68$-bus test system is considered. The synchronous generators in the system are represented by the $4^{\text{th}}$-order machine model. The generator and line parameters are obtained from \cite{nilthesis}. The system is modified to include inverter-interfaced energy storage units on all load buses. Droop-controlled GFMs \cite{REGFM_A1} act as the grid interface for these distributed storage assets. The GFMs are equipped with limiters and relays for over-current protection. The GFM parameters are obtained from \cite{REGFM_A1}.

This study is focused on analyzing the impact of GFM parameters on the inter-area oscillation modes of the system. 
The base case of the IEEE test system with only SGs, i.e., before the addition of the GFMs, has four dominant inter-area oscillation modes. These correspond to the modal frequencies $0.40$ Hz, $0.50$ Hz, $0.66$ Hz, and $0.78$ Hz, as shown in Figs \ref{fig:modal_analy_mp}. Observe that, for each mode, the damping is very poor with $\zeta < 1\%$. 
Next, we introduce the GFM-interfaced storage resources into the test system. For a fixed storage size ($10\%$ of the total system load), we sweep $m_p$ over a wide range. As we progressively decrease the droop setting $\hat{m}_{p,j}$ from a high value (i.e., the synchronous generator-only base case\footnote{The synchronous generator-only base case is equivalent to having inverter sizes $S_j \rightarrow 0$ and $m_{p_j}$ infinitely large $\forall j \in \mathcal{N}_i$ (see, (\ref{droop_def}))}) to lower values, typically prescribed by IEEE design standards, the damping of the inter-area modes increases, 
as shown in Fig. \ref{fig:modal_analy_mp}.
This verifies the claims made in Section \ref{sec:analysis}. 
\begin{figure}
    \centering
    \includegraphics[width=\linewidth]{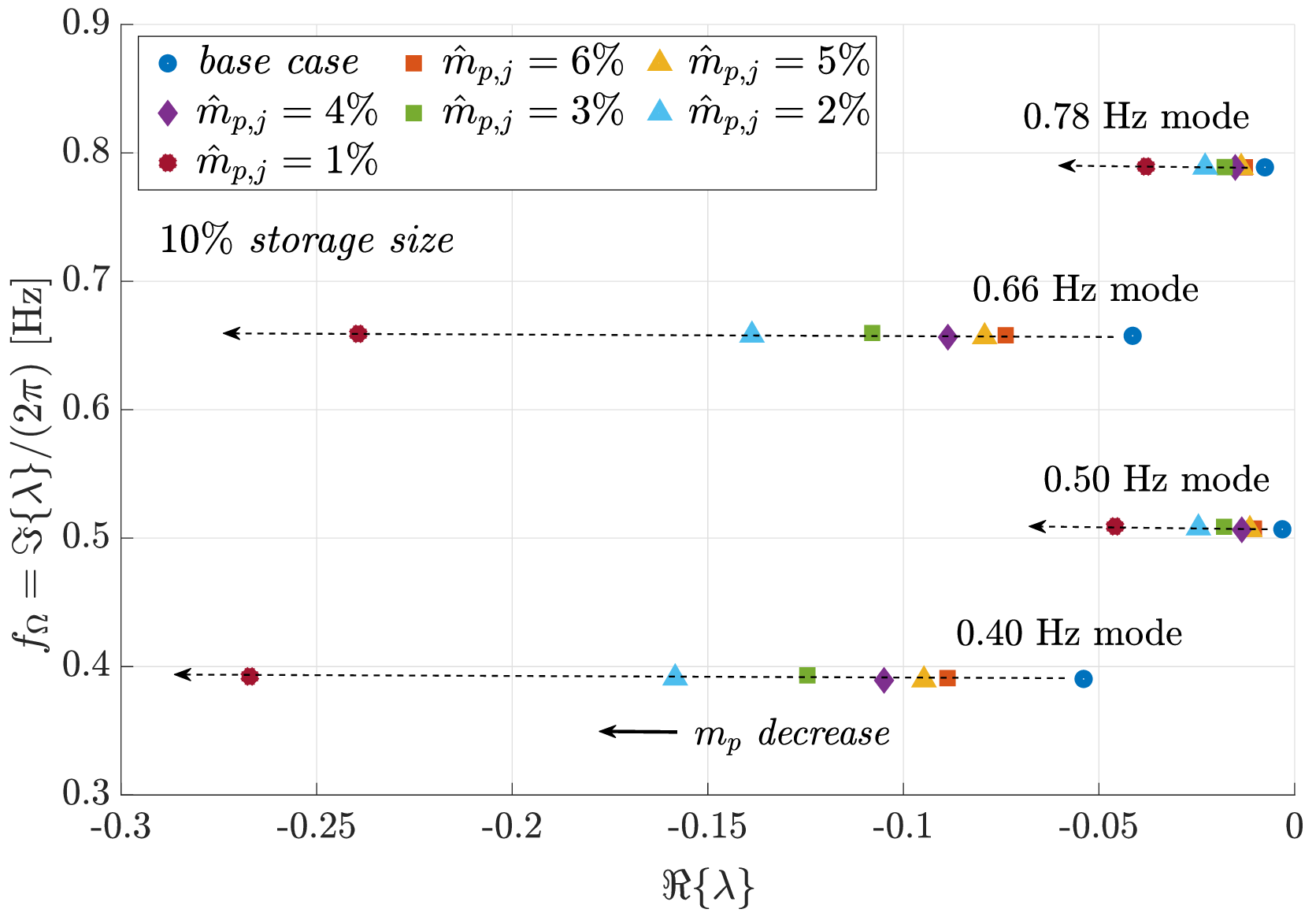}
\captionsetup{justification=raggedright,singlelinecheck=false}
    \caption{\small{Loci of the four dominant eigenvalues  (corresponding to the 4 inter-area modes) of the IEEE $68$-bus system for variation in the droop gain $\hat{m}_{p,j}$.}}
  \vspace{-0.3cm}
    \label{fig:modal_analy_mp}
\end{figure}
However, if $\hat{m}_{p,j}$ decreased to very low values, the trend in damping reverses. The critical limit on $m_p$ is attained below which the necessary conditions guaranteeing $\frac{\Re\{\Delta\lambda\}}{\Delta m_p} > 0$ do not hold. This is shown in Fig. \ref{reverse trend}. As observed, the trend is consistent for all four eigenvalues. However, it is worth noting that the $\hat{m}_{p,j}$ values in this simulation, for which this reversal of trend in damping was observed are much smaller than the typical droop setting $(2-5\%)$ prescribed in the IEEE standards \cite{Std2800}. So, it may be still reasonable to state that in the standard design range of droop coefficients, a decrease in $\hat{m}_{p,j}$ results in an increase in oscillation damping. That said, in low-inertia systems (low $M_k$ from the retirement of SGs and large $S_j$ from GFM-based storage), it may be possible for this reversal to occur at a relatively higher $\hat{m}_{p,j}$. Future work in this direction shall explore this aspect in detail. 

\begin{figure}
    \centering
\includegraphics[width=0.95\linewidth]{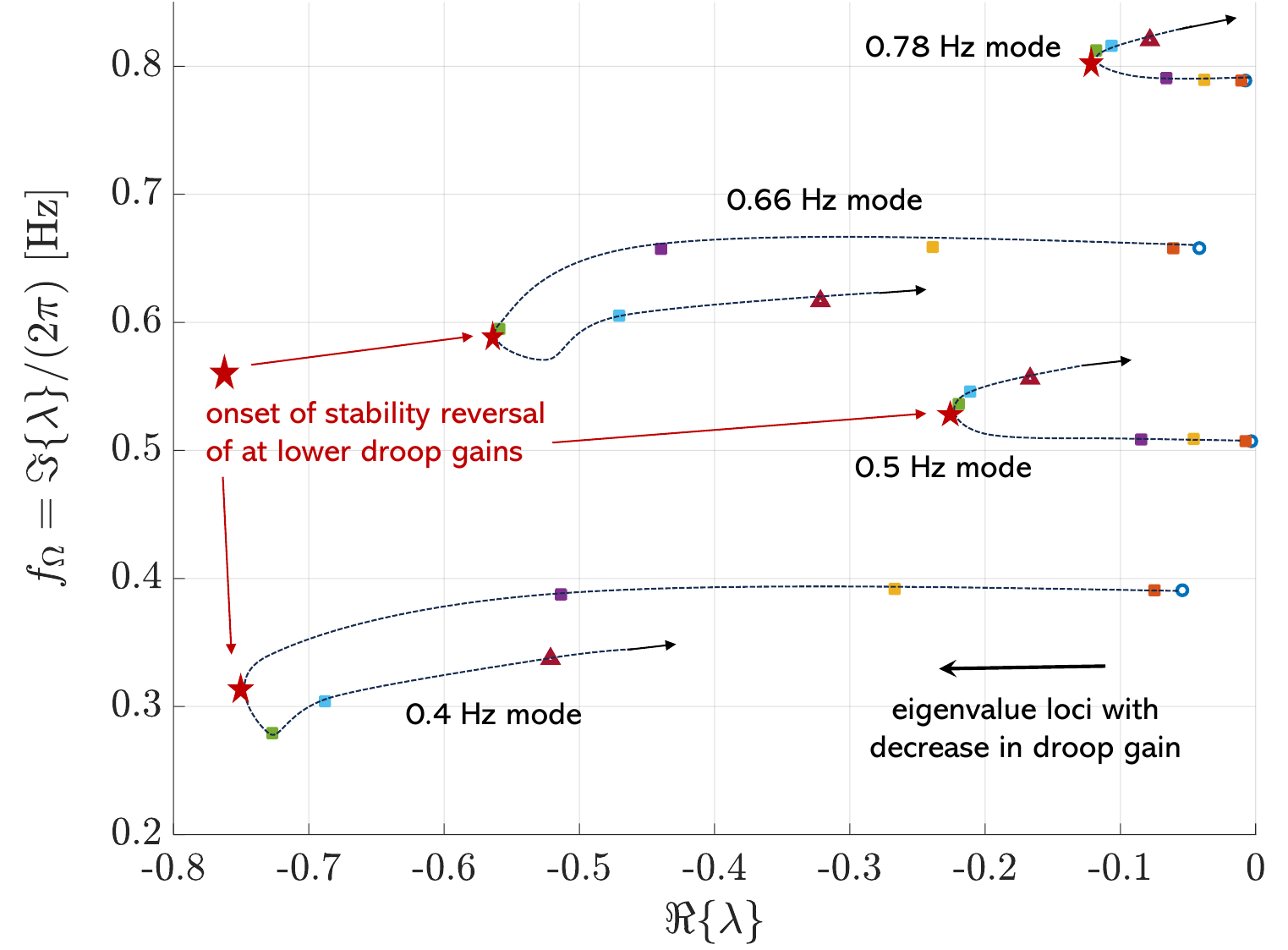}
\captionsetup{justification=raggedright,singlelinecheck=false}
    \caption{\small{Loci of four dominant eigenvalues (corresponding to the inter-area modes) for a wide range of $\hat{m}_{p,j}$ variation. In the standard range of droop gains, damping increases with a decrease in $\hat{m}_{p,j}$, however, at very low droop values the trend is reversed. }}
    \label{reverse trend}
    \vspace{-0.2cm}
\end{figure}

\begin{figure}
    \centering
    \includegraphics[width = 1.0\linewidth]{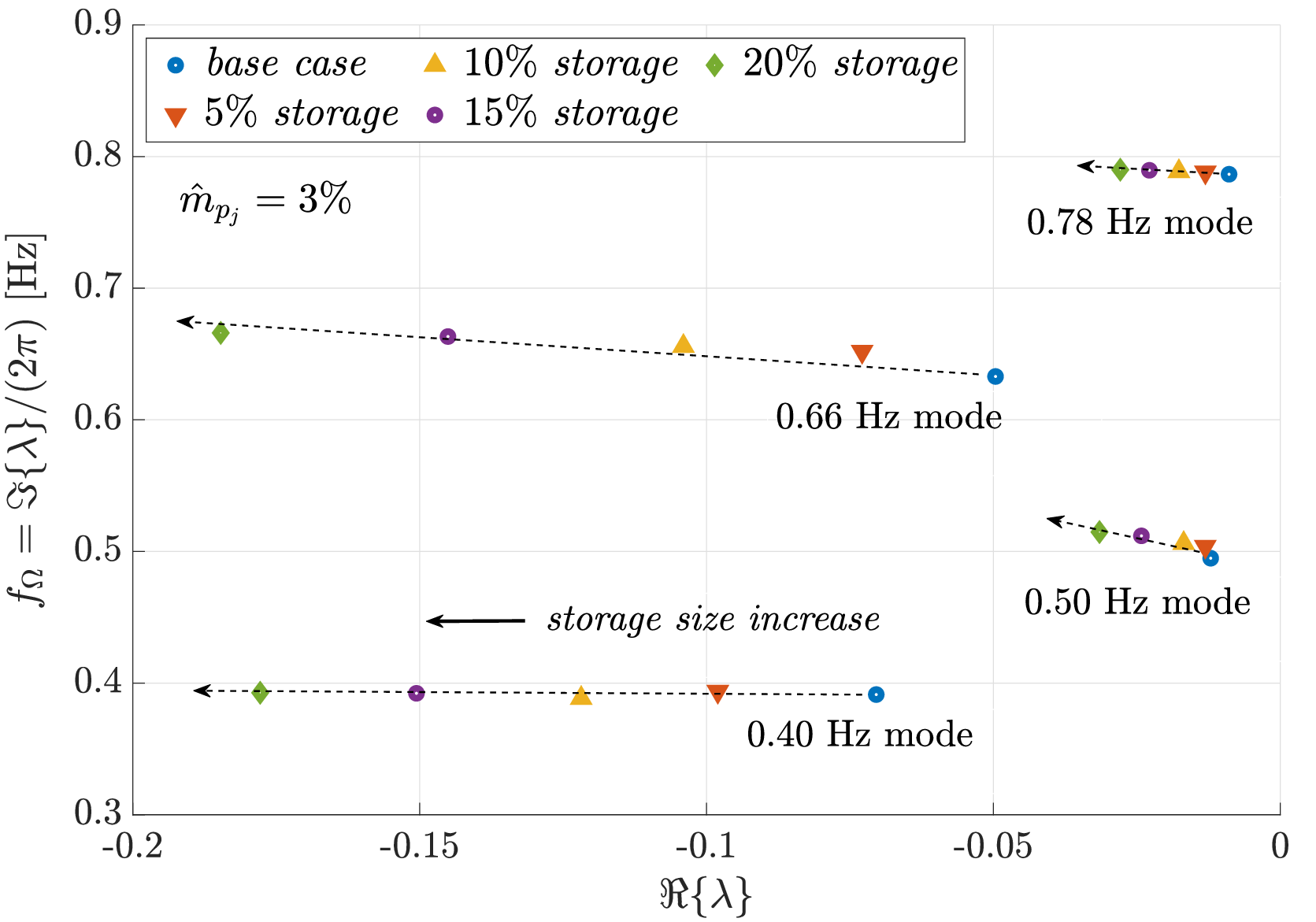}
\captionsetup{justification=raggedright,singlelinecheck=false}
    \caption{\small{Loci of four dominant eigenvalues (corresponding to the inter-area modes) with variation in the storage size.}}
    \label{fig:modal_anly_size}
    \vspace{-0.2cm}
\end{figure}

From \eqref{droop_def}, it is expected that the increase in damping observed from a decrease in $\hat{m}_{p,j}$ in Fig. \ref{fig:modal_analy_mp} should also be true for an increase in $S_j$. In other words, between two distributed storage networks with the same topology and the same droop settings on the GFMs but different storage capacities, the network with a higher storage capacity on individual inverters will have better damped inter-area modes. This is verified in Fig. \ref{fig:modal_anly_size}. Observe that, as the storage capacities of the individual inverters are gradually increased from $5\%$ of the system load to $20\%$, with the droop $m_{p,j}$ fixed at $3\%$, a significant increase in the damping is observed for all four inter-area modes.

\section{Conclusions}
\label{sec:conclusion}

This paper presented an analytical treatment for the characterization of damping in power networks with droop-based grid-forming resources. It was shown that both the capacity of the individual storage units and the droop settings on the GFMs influence the damping of the low-frequency inter-area oscillation modes. It was proved that damping of the inter-area modes increases with a decrease in the active power-frequency droop of the GFM and also with the increase in the sizes of their storage capacities. Theoretical analysis also revealed that these relationships between damping and droop setting, and storage capacity, are not monotonic. There exists a critical point beyond which decreasing the droop gains results in a decrease in damping. 

\bibliographystyle{IEEEtran}
\bibliography{Ref_damp}

\appendices
\discard{\section{Proof of Proposition\,\ref{fact1}}\label{proof:fact1}

Note that we can write $v$ and $C$ as:
\begin{align*}
    v&=\Real{v}+j\Imag{v}\,,\,~\,C ={\Real{C}}+j\,{\Imag{C}}\,,
\end{align*}
where $\Real{C}$ and $\Imag{C}$ are, respectively, symmetric and skew-symmetric. After expanding $v^H\,C\,v=0$ and performing some rearrangement of terms, we have:
\begin{align*}
    \begin{bmatrix}
        \Real{v} \\ \Imag{v}
    \end{bmatrix}^\top\Xi\,\begin{bmatrix}
        \Real{v}\\ \Imag{v}
    \end{bmatrix}&=0\,,\,\text{where}\,~\,\Xi:=\begin{bmatrix}
        \Real{C} & -\Imag{C}\\\Imag{C} & \Real{C}
    \end{bmatrix}.
\end{align*}
Note that $\Xi$ defined above is a symmetric matrix. For the above condition to hold for some $v$\,, the matrix $\Xi$ must not be either positive definite or negative definite. Therefore, a set of necessary conditions follow from the application of Schur complement (and noting that $\Real{C}$ is real symmetric):
\begin{align*}
    \lmax{\Real{C}}\ge 0\ge\lmin{\Real{C}}
\end{align*}
This completes the proof.

\section{Proof of Lemma\,\ref{fact2}}\label{proof:fact2}

Note that we can write $C$ as:
\begin{align}\label{eq:Ch}
    C &= C_h + C_{h'}\\
    \text{where, }\,~\,C_h&:=\frac{1}{2}\left(C+C^H\right)=\Symm{\Real{C}}+j\,\AntiSymm{\Imag{C}}\,.\notag\\
    C_{h'}&:=\frac{1}{2}\left(C-C^H\right)=\AntiSymm{\Real{C}}+j\,\Symm{\Imag{C}}\,,\notag
\end{align}
such that $C_h$ is Hermitian (i.e., $C_h^H=C_h$) and $C_h$ is anti-Hermitian (i.e., $C_{h'}^H=-C_{h'}$), where $j$ is used to denote the \textit{imaginary unit}, while $\Symm{\cdot}$ and $\AntiSymm{\cdot}$ denote, respectively, the symmetric and skew-symmetric parts, as defined in \eqref{eq:def_symm}. Singularity of $C$ implies that there exists a non-zero vector $v\in\mathbb{C}^{n}$ such that,
\begin{align*}
    v^H\,C\,v=0\implies v^H\,C_h\,v + v^H\,C_{h'}\,v=0
\end{align*}
Taking complex conjugate of both sides, and using the properties of Hermitian and anti-Hermitian matrices, we obtain:
\begin{align*}
    v^H\,C_h\,v - v^H\,C_{h'}\,v=0
\end{align*}
The above two equalities together yield:
\begin{align*}
    v^H\,C_h\,v = 0\,,\,~\, v^H\,C_{h'}\,v=0\,.
\end{align*}

First we focus on $v^H\,C_h\,v=0$\,, where $C_h$ is Hermitian. Note that $\Real{C_h}=\Symm{\Real{C}}$\,. Hence, on application of Proposition\,\ref{fact1}, we have the following necessary conditions:
\begin{align*}
    \lmax{\Symm{\Real{C}}}\ge 0\ge\lmin{\Symm{\Real{C}}}
\end{align*}

Next, we consider $v^H\,C_{h'}\,v=0$ and note that we can rewrite the equality as $v^H\,\left(-jC_{h'}\right)\,v=0$\,, where $-jC_{h'}$ is Hermitian. The rest of the proof follows trivially by applying Proposition\,\ref{fact1}, and using \eqref{eq:Ch}. }

\section{Proof of Proposition\,\ref{prop:intermediate_result}}\label{proof:prop_R}
Using Assumption\,\ref{AS:slower}, we can approximate:
    \begin{align*}
        &\lim_{m_p\rightarrow\infty}\left(\lambda\,\Eye_{n_i}\!\!+m_pK_{ii}\right)^{-2}\!\!\\
        &=\left(\Conj{\lambda}\,\Eye_{n_i}\!+m_p\,K_{ii}\right)^2\left(\Abs{\lambda}^2\Eye_{n_i}\!+2\Real{\lambda}m_pK_{ii}+m_p^2K_{ii}^2\right)^{-2}\!\!\\
        &\approx \left(2\,\Conj{\lambda}\,m_p\,K_{ii}+m_p^2K_{ii}^2\right)\left(m_p^2K_{ii}^2+\Order{m_p}\right)^{-2}\\
        &=\frac{1}{m_p^3}\left(2\,\Conj{\lambda}\,K_{ii}^{-3}+m_p\,K_{ii}^{-2}+\Order{1/m_p}\right)
    \end{align*}
    Applying \eqref{eq:URQTheta}, we obtain:
    \begin{align*}
        \lim_{m_p\rightarrow\infty}R&\approx \frac{1}{m_p^3}\,K_{gi}\left(2\,\Conj{\lambda}\,K_{ii}^{-3}+m_p\,K_{ii}^{-2}+\Order{1/m_p}\right)K_{ig}\\
        &=\frac{1}{m_p^3}\left(m_p\,U_1 + 2\,\Conj{\lambda}\,U_2 + \Order{1/m_p}\right)
    \end{align*}
    The expression for $\Theta$ follows similarly, when we notice that:
    \begin{align*}
        &\lim_{m_p\rightarrow\infty}\left(m_p\,\Conj{R}+\Conj{Q}(\lambda)\right)\\
        &\approx\frac{1}{m_p^2}\left(m_p\,U_1+2\,\lambda\, U_2+\Order{1/m_p}\right)+Q(\Conj{\lambda})\\
        &\approx \frac{1}{m_p}\left(U_1+m_p\,Q(\Conj{\lambda})+\Order{1/m_p}\right)\\
        \implies &\lim_{m_p\rightarrow\infty}\Theta =\lim_{m_p\rightarrow\infty} \lambda R\left(m_p\,\Conj{R}+Q(\Conj{\lambda})\right)\\
        &\approx \frac{1}{m_p^3}\lambda \left[m_p\,U_1\,Q(\Conj{\lambda})+U_1^2+2\Conj{\lambda}\,U_2\,Q(\Conj{\lambda})+\Order{1/m_p} \right]
    \end{align*}


\end{document}